\documentclass[11pt,twoside]{article}
\usepackage[margin=1in]{geometry}

\usepackage{url,hyperref,xspace}
\usepackage[usenames]{color}
\usepackage{times}
\usepackage{verbatim,makeidx}
\usepackage{algorithm}
\usepackage{algorithmic}
\usepackage{amsmath}
\usepackage{amssymb}
\usepackage{amsthm}
\usepackage{dsfont}
\usepackage{color}
\usepackage{colortbl}
\usepackage{float}
\usepackage{cite}
\usepackage{pbox}
\usepackage{bbm}
\usepackage{authblk}
\usepackage{graphicx}
\usepackage{caption}
\usepackage{subcaption}

\allowdisplaybreaks

\newcounter{ALC@tempcntr}

\theoremstyle{plain}
\newtheorem{proposition}{Proposition}
\newtheorem{lemma}{Lemma}
\newtheorem{definition}{Definition}

\theoremstyle{definition}
\newtheorem{construction}{Construction}

\theoremstyle{remark}
\newtheorem{remark}{Remark}

\newcommand{\beq}{\begin{eqnarray}}
\newcommand{\eeq}{\end{eqnarray}}

\newcommand{\field}[1]{\mathbb{#1}}

\newcommand{\F}{\field{F}}
\newcommand{\B}{\field{B}}

\newfont{\bbb}{msbm10 scaled 500}

\newfont{\bb}{msbm10 scaled 1100}

\newcommand{\av}{{\bf a}}

\newcommand{\cv}{{\bf c}}

\newcommand{\ev}{{\bf e}}
\newcommand{\fv}{{\bf f}}

\newcommand{\pv}{{\bf p}}

\newcommand{\rv}{{\bf r}}

\newcommand{\Am}{{\bf A}}

\newcommand{\Gm}{{\bf G}}
\newcommand{\Hm}{{\bf H}}
\newcommand{\Id}{{\bf I}}

\newcommand{\Cc}{{\cal C}}
\newcommand{\Dc}{{\cal D}}
\newcommand{\Ec}{{\cal E}}

\newcommand{\Mc}{{\cal M}}

\newcommand{\Rc}{{\cal R}}
\newcommand{\Sc}{{\cal S}}

\newcommand{\Yc}{{\cal Y}}

\newcommand{\remove}[1]{}

\theoremstyle{definition}

\theoremstyle{remark}

\newcommand{\latexe}{{\LaTeX\kern.125em2%
                      \lower.5ex\hbox{$\varepsilon$}}}

\chardef\bslash=`\\	

\makeatletter		
\def\square{\RIfM@\bgroup\else$\bgroup\aftergroup$\fi
\vcenter{\hrule\hbox{\vrule\@height.6em\kern.6em\vrule}
\hrule}\egroup}\makeatother\makeindex

\definecolor{OXO-emph}{RGB}{153,0,0}

\DeclareMathAlphabet{\mathpzc}{OT1}{pzc}{m}{it}
\newcolumntype{?}{!{\vrule width 1.5pt}}

\title{A Note on Secure Minimum Storage Regenerating Codes}

\author{Ankit Singh Rawat}%
\affil{Computer Science Department, \\ Carnegie Mellon University, \\Pittsburgh, 15213.\\ {E-mail:~asrawat@andrew.cmu.edu}}

\begin{document}

\maketitle



\begin{abstract} 
This short note revisits the problem of designing secure minimum storage regenerating (MSR) codes for distributed storage systems. A secure MSR code ensures that a distributed storage system does not reveal the stored information to a passive eavesdropper. The eavesdropper is assumed to have access to the content stored on $\ell_1$ number of storage nodes in the system and the data downloaded during the bandwidth efficient repair of an additional $\ell_2$ number of storage nodes. This note combines the Gabidulin codes based precoding~\cite{RKSV12} and a new construction of MSR codes (without security requirements) by Ye and Barg~\cite{YeB16a} in order to obtain secure MSR codes. Such optimal secure MSR codes were previously known in the setting where the eavesdropper was only allowed to observe the repair of $\ell_2$ nodes among a specific subset of $k$ nodes~\cite{RKSV12, GRCP13}. The secure coding scheme presented in this note allows the eavesdropper to observe repair of any $\ell_2$ ouf of $n$ nodes in the system and characterizes the secrecy capacity of linear repairable MSR codes. 
\end{abstract}


\section{Introduction}
\label{sec:intro}

Consider a distributed storage system that stores a file $\fv$ of size $\Mc$ (symbols over a finite field $\F$) on a network of $n$ storage nodes such that the file $\fv$ can be reconstructed from the content of any $k$ out of $n$ nodes in the system. In \cite{dimakis}, Dimakis et al. study the issue of recovering the content stored in a node by downloading a small amount of data from the remaining nodes in the system. This problem is referred to as the {\em node repair} problem. The ability to conduct node repair is useful in maintaining the redundancy level of the system in the event of a node failure. Moreover, the content of a temporarily unavailable node can be accessed using the rest of the (available) nodes in the system by treating the unavailable node as a failure and invoking the mechanism to repair this node. Dimakis et al. introduce {\em repair-bandwidth}, the number of symbols downloaded to repair a node, as a metric to quantify the efficiency of the node repair mechanism~\cite{dimakis}. Assuming that each node stores $\alpha$ symbols (over $\F$) and the node repair mechanism requires contacting $d \geq k$ storage nodes and downloading $\beta$ symbols from each of the contacted nodes, \cite{dimakis} presents the following fundamental trade-off between the repair-bandwidth $d\beta$ and per node storage $\alpha$.
\begin{align}
\label{eq:cutset}
\Mc \leq \sum_{i = 1}^{k}\min\big\{\alpha, (d - i + 1)\beta\big\}.
\end{align}
The codes that operate at this trade-off are referred to as {\em regenerating codes}~\cite{dimakis}. In particular, the codes corresponding to the minimum storage point, i.e., $\alpha = \frac{\Mc}{k}$, are called {\em minimum storage regenerating} (MSR) codes. Note that an MSR code is an MDS vector code~\cite{MacSlo} and operates at the following point on the trade-off defined by the bound in \eqref{eq:cutset}.
\begin{align}
\label{eq:msr}
\big(\alpha, \beta\big) =  \left(\frac{\Mc}{k}, \frac{\alpha}{d - k + 1}\right) = \left(\frac{\Mc}{k}, \frac{\Mc}{k(d - k + 1)}\right).
\end{align}
An MSR code is said to be {\em exact-repairable} if its repair mechanism ensures reconstruction of the data that is identical to the content stored on the node being repaired. The exact-repairable MSR codes form an attractive class of coding schemes as they preserve the structure of the storage system throughout the operation of the system. The problem of designing exact-repairable MSR codes has been extensively studied in \cite{RSK11, cadambe2011optimal, PapDimCad13, zigzag13, SAK15, RSE15, RKV16a, GFV16} and references therein. Recently, Ye and Barg present explicit constructions for exact-repairable MSR codes for all values of system parameters $n$, $k$ and $d$ in \cite{YeB16a, YeB16b}. These constructions enable repair of all nodes in the system as opposed to some of the earlier constructions (e.g. the constructions from \cite{zigzag13, cadambe2011optimal}) which enable bandwidth-efficient repair only for a particular set of $k$ (systematic) nodes.

In this document, we address the issue of designing distributed storage systems that protect the stored information against eavesdropping attacks. Given increasing utilization of distributed storage systems (a.k.a. cloud storage) for storing valuable and confidential information, it is important that these systems prevent leakage of the stored information to an unauthorized and (or) adversarial agent. In this paper, we present a coding scheme that is information theoretically secure against an eavesdropper who can observe the data downloaded during the repair of $\ell_2$ storage nodes and access the content stored on $\ell_1$ (additional) nodes. The secure coding scheme enables exact-repair of all nodes in the system with $d = n - 1$ and operates at the MSR point. The scheme is obtained by combining Gabidulin codes based precoding scheme~\cite{RKSV12} with a code construction presented in \cite{YeB16a}. We note that the obtained coding scheme characterizes the secrecy capacity~\cite{PRR11} of those distributed storage systems that employ linear repair mechanisms and operate at the MSR point.

\section{Background and related work}
\label{sec:background}

In this section we formally define the underlying eavesdropper model and the associated secrecy capacity. We then present a brief description two key components of our secure coding schemes: 1) Gabidulin precoding scheme and 2) a code construction from \cite{YeB16a}. We conclude the section with a discussion on the prior work in the area of designing secure coding schemes for distributed storage systems.

\subsection{System model}
\label{sec:sys}

We consider a DSS with $n$ storage nodes where each node stores $\alpha$ symbols over a finite field $\F$. We assume that the DSS employs a coding scheme such that content of any $k$ out of $n$ nodes in the system is sufficient to construct the content stored in the remaining $n - k$ nodes. Furthermore, we assume that the content of every node in the system can be {\em exactly} reconstructed by contacting any $d$ out of $(n - 1)$ remaining nodes and downloading $\beta$ symbols (over $\F$) from each of the contacted nodes. It follows from the Singleton bound that such a system can store a file with at most $k\alpha$ (independent) symbols (over $\F$). In fact, an MDS coding scheme does store a file $\fv$ of size $\Mc = k\alpha$ symbols (over $\F$). For such coding schemes, it follows from the work of Dimakis et al.~\cite{dimakis} that 
\begin{align}
\beta \geq \frac{\alpha}{d - k + 1}.
\end{align}
Here, we focus on the DSS employing those exact-repairable coding schemes that are both storage and repair-bandwidth efficient, i.e., we have that
\begin{align}
\label{eq:msr1}
(\alpha, \beta) = \Big(\frac{\Mc}{k}, \frac{\Mc}{k(d- k + 1)}\Big).
\end{align}

\subsection{Eavesdropper model and secrecy capacity}
\label{sec:eavesdropper}

We consider the $(\ell_1, \ell_2)$-eavesdropper model introduced in \cite{SRK11GCom}. Let $n$ nodes in the DSS are indexed by the set $[n] := \{1, 2,\ldots, n\}$. For $\ell_1$ and $\ell_2$ such that $\ell_1 + \ell_2 < k$, an $(\ell_1, \ell_2)$-eavesdropper can directly access the content stored on any $\ell_1$ storage nodes indexed by the set $\Ec_1 \subset [n]$. Additionally, the eavesdropper observes the data downloaded during the repair of any $\ell_2$ storage nodes indexed by the set $\Ec_2 \subset [n]$. The nodes indexed by the sets $\Ec_1$ and $\Ec_2$ are referred to as storage-eavesdropped and download-eavesdropped nodes, respectively. Note that a download-eavesdropped node may reveal more information compared to a storage-eavesdropped node as the content stored on a node is a function of the data downloaded during its repair. In this document we focus on coding schemes that are information theoretically secure against an $(\ell_1, \ell_2)$-eavesdropper. We formalize this notion in the following definition.

\begin{definition}
\label{def:secrecy}
Let $\fv^{s}$ be a secure file of size $\Mc^{s}$ symbols (over $\F$). We say that the DSS securely stores $\fv^{s}$ against an $(\ell_1, \ell_2)$-eavesdropper, if we have
$$
I(\fv^{s};\ev(\Ec_1, \Ec_2)) = 0~~~~~~\forall \Ec_1, \Ec_2 \subset [n]~\text{such that}~|\Ec_1| = \ell_1~\text{and}~|\Ec_2| = \ell_2.
$$
Here, $\ev(\Ec_1, \Ec_2)$ denotes the observations of an eavesdropped with its storage-eavesdropped and download-eavesdropped nodes indexed by the sets $\Ec_1$ and  $\Ec_2$, respectively. Equivalently, we also say that the DSS achieves a secure file size $\Mc^{s}$.
\end{definition}

\begin{remark}
\label{rem:secrecy}
The {\em secrecy capacity} of a DSS against an $(\ell_1, \ell_2)$-eavesdropper is defined as the maximum secure file size achieved by the DSS. In other words, the secrecy capacity of a DSS denotes the maximum sized secure file that it can store without leaking any information to an $(\ell_1, \ell_2)$-eavesdropper.
\end{remark}

\subsection{Preliminaries}
\label{sec:prelims}

As described in Section~\ref{sec:sys} and \ref{sec:eavesdropper}, we aim to store a secure file $\fv^{s}$ of size $\Mc^{s}$ symbols (over $\F$) in a DSS that stores $\Mc$ symbols (over $\F$) without any security guarantees. Moreover, the DSS is required to operate at the MSR point which is defined by the parameters given in \eqref{eq:msr1}. Towards this we utilize $\Mc - \Mc^{s}$ random symbols (over $\F$). The following lemma from \cite{Wyner:The75, SRK11GCom} allows us to argue that the proposed coding scheme is secure against an $(\ell_1, \ell_2)$-eavesdropper. 
\begin{lemma}[Secrecy Lemma~\cite{Wyner:The75, SRK11GCom}]
\label{lem:secrecy}
Let $\fv^{s}$ be the secure file that needs to be stored on a DSS and  $\rv$ be random symbols (independent of $\fv^{s}$). Let $\ev(\Ec_1, \Ec_2)$ be the observations of an eavesdropped with its storage-eavesdropped and download-eavesdropped nodes indexed by the sets $\Ec_1$ and  $\Ec_2$, respectively. If we have $H\big(\ev(\Ec_1, \Ec_2)\big) \leq H(\rv)$ and $H\big(\rv|\fv^{s},\ev(\Ec_1, \Ec_2)\big)=0$, then 
$$
I\big(\fv^{s};\ev(\Ec_1, \Ec_2)\big)=0.
$$
\end{lemma}

\subsubsection{Gabidulin precoding}
\label{sec:precoding}

Given a vector $\av = (a_1, a_2,\ldots, a_K) \in \F^K$ and $K$ points $\Yc = \{y_1, y_2,\ldots, y_K\} \subseteq \F^K$ which are linearly independent (over a subfield $\B$ of $\F$), the Gabidulin precoding of $\av$ is obtained in two steps.
\begin{itemize}
\item First, construct a linearized polynomial $m_{\av}(x)$ with the vector $\av$ defining its coefficients as follows.
\begin{align}
\label{eq:linearized_poly}
m_{\av}(x) = \sum_{i = 0}^{K - 1}a_{i+1}x^{|\B|^{i}} = \sum_{i = 0}^{K - 1}a_{i+1}x^{[i]},
\end{align}
where, for a positive integer $i$, we use $x^{[i]}$ to denote $x^{|\B|^i}$.
\item Evaluate the linearized polynomial $m_{\av}(x)$ at the given set of $K$ points $\Yc = \{y_1, y_2,\ldots, y_K\} \subseteq \F^K$ to obtain the associated Gabidulin precoded vector
\begin{align}
\pv(\av; \Yc) = \big(m_{\av}(y_1), m_{\av}(y_2),\ldots, m_{\av}(y_K)\big) \in \F^K.
\end{align} 
\end{itemize}

\subsubsection{Ye and Barg construction~\cite{YeB16a}}
\label{sec:YeBarg}
In \cite{YeB16a}, Ye and Barg present multiple code constructions for the MSR codes for all values of $n$, $k$ and $d$. These are the first fully explicit constructions of this nature. Here, we briefly describe one of the constructions from \cite{YeB16a} which we utilize to construct secure coding schemes at the MSR point. Similarly, other constructions from \cite{YeB16a} can also be utilized to obtain secure coding scheme.

\begin{construction}
\label{const:YeBarg}
For an element $a \in \big\{0, 1,\ldots, (n - k)^{n-1}\big\}$, $$(a_{n-1}, a_{n - 2},\ldots,a_1) \in \{0, 1,\ldots, n - k - 1\}^{n-1}$$ denotes the $(n-k)$-ary vector representation of the element $a$. For $a \in \big\{0, 1,\ldots, (n - k)^{n-1}\big\}$, $u \in \{0, 1,\ldots, n - k -1\}$ and $i \in [n-1]$, $a(i, u) \in \big\{0, 1,\ldots, (n - k)^{n-1}\big\}$ denotes the element with the following $(n-k)$-ary vector representation.
\begin{align}
(a_{n-1}, a_{n-2},\ldots, a_{i + 1}, u, a_{i - 1},\ldots, a_{1}) \in \{0, 1,\ldots, n - k - 1\}^{n-1}.
\end{align}
Let $\B$ be a field with $|\B| \geq n + 1$ and $\gamma \in \B$ be its primitive element. An MSR code $\Cc$ with $\alpha = (n - k)^{n-1}$ is defined by the following $(n-k)\alpha \times n\alpha$ parity check matrix.
\begin{align}
\label{eq:Hm}
\Hm = \left( \begin{array}{ccccc}
\Id & \Id & \cdots & \Id & \Id\\
\Am_1 & \Am_{2} & \cdots & \Am_{n-1} & \Id \\
\Am^2_1 & \Am^2_{2} & \cdots & \Am^2_{n-1} & \Id \\
\vdots & \vdots & \ddots & \vdots & \vdots \\
\Am^{n-k-1}_1 & \Am^{n-k-1}_{2} & \cdots & \Am^{n-k-1}_{n-1} & \Id \\
\end{array} \right) \in \B^{(n-k)\alpha \times n\alpha},
\end{align}
where $\Id$ denotes the $\alpha \times \alpha$ identity matrix. For $i \in [n - 1]$, $\Am_i$ is an $\alpha \times \alpha$ matrix which is defined as follows.
\begin{align}
\Am_i  = \sum_{a = 0}^{(n - k)^{n-1} - 1}\lambda_{i, a_i}\ev_{a}\ev^T_{a(i, a_i \oplus 1)} \in \B^{\alpha \times \alpha},
\end{align}
where $\oplus$ denotes addition modulo $(n - k)$, $\lambda_{i, 0} = \gamma^{i}$ and $\lambda_{i, u} = 1,~\forall~u \in [n-k-1]$. Here, $\{\ev_{a}\}_{a \in \big\{0, 1,\ldots, (n - k)^{n-1}\big\}}$ denotes the collection of $\alpha = (n - k)^{n-1}$ standard basis vectors in $\B^{\alpha}$, i.e., all but $a$-th coordinate of the vector $\ev_{a}$ are equal to zero and the $a$-th coordinate has its entry equal to $1$. 
\end{construction}

Let $\cv = (\cv_1, \cv_2,\ldots, \cv_n) \in \Cc \subseteq \B^{n\alpha}$ denote a codeword of the MSR code defined by the parity check matrix $\Hm$ (cf.~\eqref{eq:Hm}), i.e., 
\begin{align}
\Hm \cv = 0.
\end{align}
For $i \in [n]$, we have that
$$
\cv_i = \big(c_{i, 0}, c_{i, 1},\ldots, c_{i, \alpha-1}\big) \in \B^{\alpha},
$$
which denotes the $\alpha = (n - k)^{n-1}$ symbols stored on the $i$-th storage node in the system. In \cite{YeB16a}, Ye and Barg show that the code defined by $\Hm$ is an MDS array code, i.e.,
$$
|\Cc| = \big | \big\{\cv =  (\cv_1, \cv_2, \ldots, \cv_n) \in \B^{n\alpha}~:~\Hm \cv = 0 \big\} \big | = |\B|^{k\alpha}
$$
and for any codeword $(\cv_1, \cv_2, \ldots, \cv_n) \in \Cc$ and any set $\Sc = \big\{i_1, i_2,\ldots, i_{k} \big\} \subseteq [n]$, the $k\alpha$ symbols $\big(\cv_{i_1}, \cv_{i_2},\ldots, \cv_{i_k}\big)$ are sufficient to reconstruct the entire codeword $(\cv_1, \cv_2,\ldots, \cv_n)$. Furthermore, Ye and Barg establish that the code $\Cc$ is an MSR code with $d = n - 1$, i.e., for any $(\cv_1, \cv_2,\ldots, \cv_n) \in \Cc$ and $i \in [n]$, the $\alpha$ symbols stored on the $i$-th node $\cv_i$ can be reconstructed by downloading $\frac{\alpha}{n-k}$ symbols (over $\B$) from each of the remaining $n-1$ nodes. In particular, for $i \in [n-1]$, $\cv_i$ can be reconstructed by downloading the following symbols.
\begin{align}
\label{eq:node1}
\big\{c_{j, a}~:~j \neq i~\text{and}~a_i = 0\big\}.
\end{align}
Similarly, $\cv_n$ can be reconstructed by downloading the following symbols.
\begin{align}
\label{eq:node2}
\big\{c_{j, a}~:~j \neq n~\text{and}~a_1 \oplus a_2 \oplus \cdots \oplus a_{n - 1} = 0\big\}.
\end{align}
We refer the reader to \cite{YeB16a} for the further details of the construction.

\subsection{Related work}
\label{sec:prior}

Pawar et al. formally begin the study of the problem of designing coding schemes for DSS that are secure against passive eavesdropping attacks in \cite{PRR11}. For a distributed storage system that has per node storage $\alpha$ and requires downloading $\beta$ symbols from $d$ intact nodes during the repair of a failed node, Pawar at al. obtain the following upper bound on its secrecy capacity~\cite{PRR11}.
\begin{align}
\label{eq:pawar_bound}
\Mc^{s} \leq \sum_{i = \ell_1 + \ell_2 +1}^{k}\min\big\{\alpha, (d-i+1)\beta\big\}.
\end{align}
Recall that we are only considering those distributed storage systems where the content of any $k$ out of $n$ storage nodes is sufficient to reconstruct the entire stored information. In \cite{SRK11GCom}, Shah et al. utilize the product-matrix construction~\cite{RSK11} for minimum bandwidth regenerating (MBR) codes to construct coding schemes that are secure against an $(\ell_1, \ell_2)$-eavesdropper for all values of $\ell_1$ and $\ell_2$ such that $\ell_1 + \ell_2 < k$. These coding schemes operate at $\alpha = d\beta$ and attain the bound on the secrecy capacity in \eqref{eq:pawar_bound}. Shah et al. also utilize the product-matrix construction for MSR codes to design secure MSR coding schemes that achieves secure file size of 
\begin{align}
\label{eq:shahMSR}
\Mc^{s} = (k - \ell_1 + \ell_2)(\alpha - \ell_2 \beta)
\end{align}
against an $(\ell_1, \ell_2)$-eavesdropper~\cite{SRK11GCom}. Note that, for $\ell_2 \geq 1$, there is a gap between the bound in \eqref{eq:shahMSR} and the secure file size achieved in \eqref{eq:shahMSR}. In \cite{RKSV12}, Rawat et al. obtained an improved bound on the secure file size achievable at the MSR point.
\begin{align}
\label{eq:rawatMSR}
\Mc^{s} \leq \sum_{i = \ell_1 + 1}^{k - \ell}\big(\alpha - H(\Dc_{i, \Ec_2})\big),
\end{align}
where $\Ec_2$ denotes the set $\ell_2$ download-eavesdropped nodes and $\Dc_{i, \Ec_2}$ denotes the data sent by the $i$-th node for the repair of the storage nodes indexed by the set $\Ec_2$. Furthermore, for linear repair schemes with $d = n - 1$ and $\ell_2 \leq 2$, the bound in \eqref{eq:rawatMSR} specializes to the following~\cite{RKSV12}.
\begin{align}
\label{eq:rawatMSR2}
\Mc^{s} \leq (k - \ell_1 + \ell_2)\left(1 - \frac{1}{n-k}\right)^{\ell_2}\alpha.
\end{align}
In \cite{GRCP13}, Goparaju et al. show that the bound in \eqref{eq:rawatMSR2} holds for all values of $\ell_2$. They further generalize this bound and show that for linear repair schemes with $k \leq d \leq n - 1$, the secure file size achievable at the MSR point satisfies the following~\cite{GRCP13}.
\begin{align}
\label{eq:goparajuMSR}
\Mc^{s} \leq (k - \ell_1 + \ell_2)\left(1 - \frac{1}{d-k+1}\right)^{\ell_2}\alpha.
\end{align}
As for the achievability schemes, for $d = n - 1$, Rawat et al. obtain a secure coding scheme at the MSR point that attain the bound in \eqref{eq:rawatMSR2} provided that, whenever $\ell_2 \geq 2$, download-eavesdropped nodes are restricted to a fixed set of $k$ nodes among the $n$ nodes in the system. This coding scheme is obtained by combining the Gabidulin precoding (cf.~Section~\ref{sec:precoding}) with the zigzag codes from \cite{zigzag13}. We also note that for $\ell_2 = 1$, the secure coding scheme from \cite{SRK11GCom} is optimal as it attains the bound in \eqref{eq:rawatMSR}. Recently, Huang et al. further explore the problem of characterizing the secrecy capacity of MSR codes in \cite{HPX15a}. For $\beta = 1$, the secure files size in \eqref{eq:shahMSR} is shown to be optimal~\cite{HPX15a, SKSRR14}. Huang et al. show that optimality of the bound in \eqref{eq:shahMSR} for the MSR codes with the bounded values of $\beta$. The problem of obtaining bounds on the secrecy capacity of distributed storage systems is also studied in \cite{Tandon16} under the non black-box version of the problem. For brevity, we skip a discussion on this and refer the reader to \cite{Tandon16, GRC15}.

In this paper, we establish that for linear repairable DSS with $d = n - 1$, the bound in \eqref{eq:rawatMSR2} is the exact characterization of the secrecy capacity of an MSR code. One of the codes constructions of MSR codes from \cite{YeB16a} (cf.~Section~\ref{sec:YeBarg}) enables us to remove the restriction appearing in the secure coding scheme from \cite{RKSV12} that the download-eavesdropped nodes be restricted to a subset of $k$ nodes. As for the possibility of attaining a larger secure file size by utilizing non-linear repair schemes, Goparaju et al. show Pareto optimality of the linear repairable MSR codes among those MSR codes that simultaneously allow for all values of $\ell_2$ during the design of a secure coding scheme operating at the MSR point~\cite{GRC15}.

The cooperative regenerating codes enable simultaneous bandwidth efficient repair of multiple node failures~\cite{ShumHu, Kermarrec:Repairing11}. Security of DSS employing cooperative regenerating codes against passive eavesdropping attacks is explored in \cite{KRV12, HPX15b}. Locally repairable codes (LRCs) is another class of codes designed to be employed in distributed storage systems~\cite{Gopalan12, PapDim12}. These codes aim at repairing a failed node by contacting a small number of surviving nodes in the system. We note that the problem of designing secure locally repairable codes against passive eavesdropping attacks is considered in \cite{RKSV12, AgaMaz15}.


\section{Secure MSR codes}
\label{sec:secureMSR}

In this section we present a linear repairable coding scheme that operates at the MSR point with $d = n - 1$ and achieve the optimal secure file size in this setting.

\begin{construction}
\label{const:secureMSR}
Let $n$ and $k$ be given system parameters. Let $\fv^{s}$ be a secure file of size 
\begin{align}
\label{eq:filesecure}
\Mc^{s} = (k - \ell_1 - \ell_2) \left(1 - \frac{1}{n - k}\right)^{\ell}(n - k)^{n-1} 
\end{align}
symbols (over a finite field $\F$). We assume that we have $\F = \B^{Q}$, where $Q \geq k(n - k)^{n-1}$ (cf.~\eqref{eq:filesecure}) and $|\B| \geq n + 1$. Let $\alpha = (n - k)^{n-1}$ and $\Mc = k\alpha = k(n - r)^{n-1}$. We now generate a coding scheme that securely stores the file $\fv^{s}$ against an $(\ell_1, \ell_2)$-eavesdropped in the following two step process.
\begin{enumerate}
\item Let $\rv$ denote $\Mc - \Mc^{s}$ i.i.d. random symbols (independent of $\fv^{s}$) that are uniformly distributed over $\F$. We take $\Mc$ linearly independent point (over $\B$) $\Yc = \big\{y_1, y_2,\ldots, y_{\Mc}\big\} \subset \F$ and perform Gabidulin precoding of the vector $\av = (\ev, \fv^{s}) \in \F^{\Mc}$ as defined in Section~\ref{sec:precoding}. Let $\fv$ denote the precoded vector, i.e., 
\begin{align}
\fv = \pv(\av, \Yc) = \big(m_{\av}(y_1), m_{\av}(y_2),\ldots, m_{\av}(y_{\Mc})\big) \in \F^{\Mc}, 
\end{align}
where $m_{\av}(x)$ is the linearized polynomial associated with the vector $\av = (\rv, \fv^{s})$ as define in \eqref{eq:linearized_poly}, i.e.,
\begin{align}
\label{eq:linear_poly}
m_{\av}(x) = \sum_{i = 0}^{\Mc - 1}a_{i+1}x^{[i]}.
\end{align}
\item Let $\Gm$ be a $k\alpha \times n\alpha$ generator matrix for the MSR code $\Cc$ with $d = n - 1$ and $\alpha = (n - k)^{n-1}$ obtained by Construction~\ref{const:YeBarg} (cf.~Section~\ref{sec:YeBarg}), i.e.,
\begin{align}
\Gm\Hm^{T} = \mathbf{0},
\end{align}
where $\mathbf{0}$ denotes the $k\alpha$-length zero vector. Given the precoded vector $\fv$ from the previous stage, we obtain the associated code vector in $\Cc$ as follows.
\begin{align}
\label{eq:codeword}
(\cv_1, \cv_2,\ldots, \cv_n) = \fv \cdot \Gm \in \Cc \subseteq \F^{n\alpha}.
\end{align}
\end{enumerate}
For $i \in [n]$, the $i$-th storage node stores the $\alpha = (n - k)^{n-1}$ symbols in the subvector $\cv_i$ (cf.~\eqref{eq:codeword}).
\end{construction}

The repairability of the proposed coding scheme with the repair-bandwidth $\Big(\frac{n-1}{n-k}\Big)\cdot\alpha$ symbols (over $\F$) follows from the repairability of the code $\Cc$ (cf.~\ref{sec:YeBarg}).  Next, we argue that the proposed coding scheme is secure against an $(\ell_1, \ell_2)$-eavesdropper. Towards this, we present the following simple lemma.

\begin{lemma}
\label{lem:simple}
Let $\Sc \subseteq [n]$. For $j \in [n]\backslash \Sc$, let $\Dc_{j, \Sc}$ denote the symbols downloaded from the $j$-th storage node during the repair of the storage nodes indexed by the set $\Sc$. Then, we have
\begin{align}
\big|\Dc_{j, \Sc}\big| = \left(1 - \left(1 - \frac{1}{n - k}\right)^{|\Sc|}\right)\cdot (n - k)^{n-1}.
\end{align}
\end{lemma}
\begin{proof}
We divide the proof in two cases:
\begin{itemize}
\item {\bf Case~$1$~($n \notin \Sc$):}~It follows from \eqref{eq:node1} that
\begin{align}
\big|\Dc_{j, \Sc}\big| &= \Big |\bigcup_{i \in \Sc} \big\{c_{j, a}~;~a_{i} = 0\big\}\Big|  \nonumber \\
& = (n  - k)^{n - 1} - \Big |\bigcap_{i \in \Sc} \big\{c_{j, a}~;~a_{i} \neq 0\big\}\Big| \nonumber \\
& = (n - k)^{n - 1} - (n - k)^{n - 1 - |\Sc|}\cdot (n - k - 1)^{|\Sc|} \nonumber \\
&=  \left(1 - \left(1 - \frac{1}{n - k}\right)^{|\Sc|}\right)\cdot (n - k)^{n-1}.
\end{align}
\item {\bf Case~$2$~($n \in \Sc$):}~Let's define $\widetilde{\Sc} = \Sc \backslash \{n\}$. It follows from \eqref{eq:node1} and \eqref{eq:node2} that
\begin{align}
\big|\Dc_{j, \Sc}\big| &= \Big |\bigcup_{i \in \widetilde{\Sc}} \big\{c_{j, a}~;~a_{i} = 0\big\} \cup \big\{c_{j, a}~:~a_1 \oplus a_2 \oplus \cdots \oplus a_{n - 1} = 0\big\}\Big|  \nonumber \\
& = (n  - k)^{n - 1} - \Big |\bigcap_{i \in \widetilde{\Sc}} \big\{c_{j, a}~;~a_{i} \neq 0\big\} \cap \big\{c_{j, a}~:~a_1 \oplus a_2 \oplus \cdots \oplus a_{n - 1} \neq 0\big\} \Big| \nonumber \\
& = (n - k)^{n - 1} - (n - k)^{n - 1 - |\widetilde{\Sc}|}\cdot (n - k - 1)^{|\widetilde{\Sc}|}\cdot \left(\frac{n - k + 1}{n - k}\right) \nonumber \\
&=  \left(1 - \left(1 - \frac{1}{n - k}\right)^{|\Sc|}\right)\cdot (n - k)^{n-1}.
\end{align}
\end{itemize}
This completes the proof.
\end{proof}
\begin{proposition}
\label{prop:secure}
The coding scheme described in Construction~\ref{const:secureMSR} is secure against an $(\ell_1, \ell_2)$-eavesdropper.
\end{proposition}
\begin{proof}
The proof of this proposition is very similar to the proof of \cite[Theorem 18]{RKSV12}. Let $\Ec_1$ and $\Ec_2$ denote the indices of the storage-eavesdropped and download-eavesdropped nodes, respectively. Let's consider a set of $k - |\Ec_1| - |\Ec_2| = k - \ell_1 - \ell_2$ storage nodes $\Rc$ such that $\Rc \cap \big\{\Ec_1 \cup \Ec_2\big\} = \emptyset$. Let $\ev\big(\Ec_1, \Ec_2\big)$ denote the symbols observed by the eavesdropper. Note that
\begin{align}
\label{eq:eve2}
\ev\big(\Ec_1, \Ec_2\big) = \big\{\cv_i~:~i \in \Ec_1\big\} \bigcup  \Big\{ \bigcup_{i \in \Ec_2}\big\{\cup_{j \neq i} \Dc_{j, i} \big\} \Big \}.
\end{align}
Consider 
\begin{align}
\label{eq:cond1}
H\big( \ev\big(\Ec_1, \Ec_2\big) \big) &= H\Big(\big\{\cv_i~:~i \in \Ec_1\big\} \bigcup  \Big\{ \bigcup_{i \in \Ec_2}\big\{\cup_{j \neq i} \Dc_{j, i} \big\} \Big\}\Big) \nonumber \\
&\overset{(a)}{=} H\Big(\big\{\cv_i~:~i \in \Ec_1\cup \Ec_2\big\} \bigcup  \Big\{ \bigcup_{i \in \Ec_2}\big\{\cup_{j \neq i} \Dc_{j, i} \big\} \Big\}\Big) \nonumber \\
&= H\left(\big\{\cv_i~:~i \in \Ec_1\cup \Ec_2\big\}\right) + H\left( \bigcup_{i \in \Ec_2}\big\{\cup_{j \neq i} \Dc_{j, i} \big\} \Big\} \Big| \big\{\cv_i~:~i \in \Ec_1\cup \Ec_2\big\} \right) \nonumber \\
&\overset{(b)}{=} (\ell_1 + \ell_2)\alpha + H\Big(\bigcup_{i \in \Ec_2}\big\{\cup_{j \in \Rc} \Dc_{j, i} \big\} \Big\}\Big) \nonumber \\
& = (\ell_1 + \ell_2)\alpha + H\Big(\cup_{j \in \Rc}\Dc_{j, \Ec_2}\Big) \nonumber \\
& \leq (\ell_1 + \ell_2)\alpha + \sum_{j \in \Rc}H\big(\Dc_{j, \Ec_2}\big) \nonumber \\
& \leq  (\ell_1 + \ell_2)\alpha + \sum_{j \in \Rc}\big|\Dc_{j, \Ec_2}\big| \nonumber \\
&\overset{(c)}{=} (\ell_1 + \ell_2)\cdot (n -k)^{n-1} + (k - \ell_1 - \ell_2)\cdot \left(1 - \left(1 - \frac{1}{n - k}\right)^{\ell_2}\right)\cdot (n - k)^{n-1} \nonumber \\
& = k\cdot(n-k)^{n-1} - (k - \ell_1 - \ell_2)\cdot\left(1 - \frac{1}{n - k}\right)^{\ell_2}\cdot (n - k)^{n-1} \nonumber \\
&= \Mc - \Mc^{s} = H(\rv).
\end{align}
where step $(a)$ follows from the fact that for $i \in \Ec_2$, $\cv_i$ is a function of the symbols in the set $\big\{\cup_{j \neq i} \Dc_{j, i} \big\}$. The steps $(b)$ and $(c)$ follow from \cite[Lemma~5]{HPX15a} and Lemma~\ref{lem:simple}, respectively.

Since $\Gm$ is a generator matrix of an MDS array code, it follows from \cite[Lemma 9]{RKSV12} that the symbols in the set 
\begin{align}
\label{eq:eve3}
\{\cv_i~:~i \in \Ec_1\cup \Ec_2\big\} \bigcup \Big\{\cup_{j \in \Rc}\Dc_{j, \Ec_2}\Big\}
\end{align}
correspond to the evaluations of the linearized polynomial $m_{\av}(x)$ (cf.~\eqref{eq:linear_poly}) at  
$$
\big|\{\cv_i~:~i \in \Ec_1\cup \Ec_2\big\} \bigcup \Big\{\cup_{j \in \Rc}\Dc_{j, \Ec_2}\Big\} \big| = \Mc - \Mc^{s}
$$
linearly independent (over $\B$) points in $\F$. Note that the symbols in \eqref{eq:eve2} can be obtained from the symbols observed by the eavesdropper $\ev\big(\Ec_1, \Ec_2\big)$ (cf.~\eqref{eq:eve2}). Given the secure file $\fv^{s}$, one can remove the contribution of $\fv^{s}$ from these evaluations of $m_{\av}(x)$ to obtain the $\Mc - \Mc^{s}$ evaluations of the following polynomial at the $\Mc - \Mc^{s}$ linearly independent (over $\B$) points in $\F$. 
\begin{align}
\label{eq:eve_poly}
m_{\rv}(x) = \sum_{i = 0}^{\Mc - \Mc^{s} - 1} a_{i+1}x^{[i]} = \sum_{i = 0}^{\Mc - \Mc^{s} - 1} r_{i+1}x^{[i]},
\end{align}
where the last equality holds as the first $\Mc - \Mc^{s}$ coordinates of the vector $\av$ are composed of $\Mc - \Mc^{s}$ random symbols $\rv$ (cf.~Construction~\ref{const:secureMSR}). Now, it is straightforward from \cite[Remark~8]{RKSV12} that these evaluations are sufficient to recover the coefficients of the linearized polynomial $m_{\rv}(x)$. In other words, we have that 
\begin{align}
\label{eq:cond2}
H\big(\rv~|~\fv^{s}, \ev(\Ec_1, \Ec_2)\big) = 0.
\end{align}
It follows from \eqref{eq:cond1} and \eqref{eq:cond2} that the coding scheme defined in Construction~\ref{const:secureMSR} satisfies the both requirements of Lemma~\ref{lem:secrecy}. Thus, we have 
$$
I\big(\fv^{s}; \ev(\Ec_1, \Ec_2)\big) = 0.
$$
Since the choice of $\Ec_1$ and $\Ec_2$ is arbitrary, this establishes that the coding scheme obtained by Construction~\ref{const:secureMSR} is secure against an $(\ell_1, \ell_2)$-eavesdropper.
\end{proof}
\section{Conclusion}
\label{sec:conclusion}

We characterize the secrecy capacity of linear repairable MSR codes with $d = n -1$ against a passive eavesdropping attack, where the eavesdropper is allowed to observe repair of $\ell_2$ storage nodes in addition to the content stored on $\ell_1$ storage nodes. One of the code constructions for MSR codes from \cite{YeB16a} proves instrumental in establishing this result. It is an interesting question to establish the similar results for general values of $d \in \{k, k+1,\ldots, n-1\}$. Another direction for future work is to characterize the secrecy capacity of minimum storage cooperative regenerating (MSCR) codes.


\bibliographystyle{plain}
\bibliography{SecureMSR_note}


\end{document}